\documentclass{article}
\usepackage{arxiv}

\usepackage{amsthm}

\usepackage{xcolor}  

\usepackage{amssymb}

\usepackage{amsmath}
\usepackage{mathdots}

\usepackage{mathtools}

\usepackage{tensor}

\usepackage{urwchancal}
\DeclareFontFamily{OT1}{pzc}{}
\DeclareFontShape{OT1}{pzc}{m}{it}{<-> s * [1.10] pzcmi7t}{}
\DeclareMathAlphabet{\mathpzc}{OT1}{pzc}{m}{it}

\usepackage{graphicx}
\usepackage{ifpdf}
\ifpdf
\usepackage{epstopdf}
\PrependGraphicsExtensions{.svg}
\fi

\newtheorem{theorem}{Theorem}[section]

\newtheorem{corollary}[theorem]{Corollary}
\newtheorem{proposition}[theorem]{Proposition}

\newtheorem{remark}[theorem]{Remark}

\providecommand{\R}{\mathbb{R}}

\providecommand{\SO}{\mathbf{SO}}

\providecommand{\GL}{\mathbf{GL}}
\providecommand{\SE}{\mathbf{SE}}

\providecommand{\grpG}{\mathbf{G}}

\providecommand{\gothg}{\mathfrak{g}}

\providecommand{\gothX}{\mathfrak{X}} 

\providecommand{\se}{\mathfrak{se}}

\providecommand{\Sph}{\mathrm{S}}

\providecommand{\calM}{\mathcal{M}}
\providecommand{\calN}{\mathcal{N}}

\providecommand{\vecL}{\mathbb{L}}

\providecommand{\vecV}{\mathbb{V}}
\providecommand{\vecW}{\mathbb{W}}

\providecommand{\tT}{\mathrm{T}} 

\providecommand{\Id}{I} 

\DeclareMathOperator{\diag}{diag}
\DeclareMathOperator{\stab}{stab}
\DeclareMathOperator{\Ad}{Ad}

\DeclareMathOperator*{\argmin}{argmin}

\providecommand{\pr}{\mathbb{P}} 

\providecommand{\td}{\mathrm{d}}
\providecommand{\tD}{\mathrm{D}}

\providecommand{\ddt}{\frac{\td}{\td t}}

\providecommand{\mr}[1]{\mathring{#1}} 

\usepackage{accents}
\makeatletter
\providecommand{\scirc}{%
    \hbox{\fontfamily{\rmdefault}\fontsize{0.4\dimexpr(\f@size pt)}{0}\selectfont{\raisebox{-0.52ex}[0ex][-0.52ex]{$\circ$}}}}

\makeatother

\mathchardef\mhyphen="2D

\usepackage{graphicx}      
\usepackage{natbib}        
\usepackage{hyperref}
\usepackage{algorithm,algpseudocode}

\providecommand{\dds}{\frac{\td}{\td s}}

\newcommand{\papercitation}{
Hampsey, M., van Goor, P., Hamel, T., Mahony, R. (2022), ``Tracking control on homogeneous spaces: the Equivariant Regulator (EqR)''.
}

\newcommand{\publicationdetails}
{
\copyrightNoticeIFACAccepted{2023} 
\papercitation 
}
\newcommand{\publicationversion}
{Author accepted version}

\begin{document}

\title{Tracking control on homogeneous spaces: the Equivariant Regulator (EqR)}
\headertitle{Tracking control on homogeneous spaces: the Equivariant Regulator (EqR)}

\author{
\href{https://orcid.org/0000-0001-2345-6789}{\includegraphics[scale=0.06]{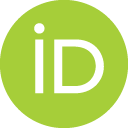}\hspace{1mm}
Matthew Hampsey}
\\
	Systems Theory and Robotics Group \\
	Australian National University \\
    ACT, 2601, Australia \\
	\texttt{matthew.hampsey@anu.edu.au} \\
	\And	\href{https://orcid.org/0000-0001-2345-6789}{\includegraphics[scale=0.06]{orcid.png}\hspace{1mm}
Pieter van Goor}
\\
	Systems Theory and Robotics Group \\
	Australian National University \\
    ACT, 2601, Australia \\
	\texttt{pieter.vangoor@anu.edu.au} \\
\And	\href{https://orcid.org/0000-0001-2345-6789}{\includegraphics[scale=0.06]{orcid.png}\hspace{1mm}
Robert Mahony}
\\
	Systems Theory and Robotics Group \\
	Australian National University \\
	ACT, 2601, Australia \\
	\texttt{robert.mahony@anu.edu.au} \\
}

\maketitle

\begin{abstract}
	Accurate tracking of planned trajectories in the presence of perturbations is an important problem in control and robotics.
	Symmetry is a fundamental mathematical feature of many dynamical systems and exploiting this property offers the potential of improved tracking performance.
    In this paper, we investigate the tracking problem for input-affine systems on homogeneous spaces: manifolds which admit symmetries with transitive group actions.
	We show that there is natural manner to lift any desired trajectory of such a system to a lifted trajectory on the symmetry group.
	This construction allows us to define a global tracking error and apply LQR design to obtain an approximately optimal control in a single coordinate chart.
	The resulting control is then applied to the original plant and shown to yield excellent tracking performance.
	We term the resulting design methodology the \emph{Equivariant Regulator} (EqR).
	We provide an example input-affine system posed on a homogeneous space, derive the trajectory linearisation in error coordinates and demonstrate the effectiveness of EqR compared to standard approaches in simulation.
	\end{abstract}



\section{Introduction}
A fundamental control task is the robust tracking of a desired state space trajectory subject to disturbances.
This problem has been extensively studied and has lead to a range of design paradigms (\cite{slotine_1983_trajectory}, \cite{Yang_1999_sliding_traj}).
A common approach for systems on $\R^m$ is to apply a Linear Quadratic Regulator (LQR) to Euclidean error dynamics linearised around the desired trajectory \cite[chapter 4]{anderson2007optimal}.

In general, the state space for a control system is a smooth manifold and error coordinates are constructed by taking the Euclidean error of the trajectories in local coordinates (for example, \cite{suicmez2014optimal}) or, in the case of an embedded manifold, by taking the error of the trajectories with respect to the Euclidean structure of the ambient space and then projecting this onto the trajectory tangent space(for example, \cite{2018_Scaramuzza_onboard_LQR}).
The former of these approaches can introduce artifacts like singularities and requires chart switching as the system state moves around the manifold.
The latter approach requires an embedding structure compatible with the geometry of the manifold, but is also ignorant of the global manifold topology and geometry and can lead to unexpected behaviour when the system state is far from the desired point.

If the system manifold admits a Lie group structure, a global, intrinsic error can be formulated in the group and used as the basis for control design.
This approach was originally used in \cite{Meyer1971} to design an almost-global controller on $\SO(3)$ for spacecraft attitude control.
This same approach has been subsequently used in multiple controller designs in robotics and aerospace applications in the intervening decades (\cite{Wie1985}, \cite{2010_Lee_Geometric}).
The error can be mapped into the Lie algebra by taking the logarithm and used for non-linear controllers (\cite{Johnson2021}) or for linear optimal controllers by mapping the error into local coordinates via the identification of the Lie algebra with $\R^m$ (\cite{2019_Farrell_Error}, \cite{2020_Forbes_Quadrotor}, \cite{Hampsey2022}).

There are, however, systems important to robotics and control for which this approach is not applicable: for example, the two-dimensional sphere $\Sph^2$ is well-known to not admit any Lie group structure (\cite{poincare1885courbes}).
Recently, development in the theory of equivariant filters has extended symmetry-based design principles to the general setting of homogeneous manifolds, where a global intrinsic error can be formed directly on the manifold (\cite{2020_mahony_EquivariantSystems}, \cite{EQF2022}).

In this paper, we propose the Equivariant Regulator (EqR): an extension of the developments in symmetry-based design on homogeneous spaces to the tracking problem.
We show that for any desired trajectory on a homogeneous space, there exists a lifted trajectory on the symmetry group.
This construction allows for the definition of a global coordinate-free tracking error for input-affine systems.
This tracking error is centered on a single point, allowing for freedom in the choice of a single coordinate chart.
We derive the dynamics of this error in coordinates and provide the linearisation in terms of the lift function.
These linearised error dynamics are used as the basis of an LQR design, yielding tracking control inputs.
The resulting control is then applied to the original plant.
To empirically demonstrate the performance of the approach, we apply the methodology to an example input-affine system on the homogeneous space $\Sph^2 \times \R^3 \times \R^3$.
We show that the EqR exhibits improved robustness with respect to a state-of-the-art LQR design in the presence of perturbations and particularly for large initialisation errors.
We claim this is due to the global error parametrisation and the manner in which the EqR handles large errors.
This paper extends the recent paper by the authors (\cite{Hampsey2022}) from free group actions, where the symmetry group is in one-to-one correspondence with the state, to symmetries on general homogeneous spaces.

\section{Preliminaries}

\subsection{Notation}
For a thorough introduction to smooth manifolds and Lie group theory, the authors recommend \cite{2012_lee_SmoothManifolds} or \cite{tu2010introduction}.

Let $\vecV$ and $\vecW$ be vector spaces. For a function $g : \vecV \to \vecW$, the notation $g[v]$ will be used to indicate that $g$ is a linear function of $\vecV$.

Let $\calM$ and $\calN$ denote smooth manifolds.
For an arbitrary point $\xi \in \calM$, the tangent space of $\calM$ at $\xi$ is denoted by $\tT_{\xi}\calM$.
For a smooth function $h : \calM \to \calN$ the notation
\begin{align*}
	\tD _{\xi|\zeta} h(\xi):\tT_{\zeta}\calM\to \tT_{h(\zeta)}\calN \\
	\delta \mapsto \tD_{\xi|\zeta}h(\xi)[\delta]
\end{align*}
denotes the differential of $h(\xi)$ evaluated at $\xi = \zeta$ in the direction $\delta \in \tT_{\zeta}\calM$.
When the basepoint and argument are implied the notation $\tD h$ will also be used for simplicity.
The space of smooth vector fields on $\calM$ is denoted with $\gothX(\calM)$.

Let $\grpG$ denote an arbitrary Lie group and denote the identity element with $\Id$.
The \emph{Lie algebra of $\grpG$} is identified with the tangent space of $\grpG$ at identity, $\gothg \simeq \tT_{\Id}\grpG$.
Given arbitrary $X \in \grpG$, \emph{left translation by $X$} is defined by
\mbox{$\mathrm L_X : \grpG \to \grpG$}, \mbox{$\mathrm L_X(Y) = XY.$}
This induces a corresponding function on $\gothg$,
\mbox{$\tD \mathrm L_X : \gothg \to \tT_X \grpG,$}
which is also referred to as left translation by $X$.
Similarly, \emph{right translation by $X$} is defined by
\mbox{$\mathrm R_X : \grpG \to \grpG$}, \mbox{$\mathrm R_X(Y) = YX.$}
This also induces a corresponding right translation on $\gothg$,
\mbox{$\tD \mathrm R_X : \gothg \to \tT_X \grpG.$}
If $\grpG \subset \GL(m, \R)$ then $\gothg \subset \R^{m \times m}$	and \mbox{$\tD \mathrm L_{X} U = XU$} and \mbox{$\tD \mathrm R_{X}U = UX$} are given by matrix multiplication.
Given a matrix Lie group $\grpG$ with Lie algebra $\gothg$, the matrix exponential $\exp$ is a local diffeomorphism between $\gothg$ and $\grpG$ when restricted to a small enough neighborhood of $0 \in \gothg$.

Given $X \in \grpG$, the adjoint map is defined by \mbox{$\Ad_X: \gothg \to \gothg$}, \mbox{$\Ad_X(U) = \tD \mathrm L_X \tD \mathrm R_{X^{-1}}U$}.
If $\grpG$ is a matrix Lie group, then \mbox{$\Ad_X(U) = XUX^{-1}$}.

Let $\grpG$ be a Lie group and $\calM$ a smooth manifold.
A left action is a smooth function $\phi: \grpG \times \calM \to \calM$ satisfying the \emph{identity} and \emph{compatibility } properties:
\begin{align*}
	\phi(\Id, \xi)        & = \xi,           &
	\phi(X, \phi(Y, \xi)) & = \phi(XY, \xi),
\end{align*}
for all $\xi \in \calM$ and $X \in \grpG$.
Given a point $X \in \grpG$, the partial map $\phi_X : \calM \to \calM$, $\phi_X(\xi) = \phi(X, \xi)$ can be formed by fixing the first argument.
Similarly, by fixing the second argument, the partial map $\phi_\xi: \grpG \to \calM$, $\phi_\xi(X) = \phi(X, \xi)$ is formed.
A group action is called \emph{transitive} if for each pair $\xi, \zeta \in \calM$, there exists a $X \in \grpG$ such that $\phi(X, \xi) = \zeta$.
If a manifold $\calM$ admits a transitive group action then it is called a \emph{homogeneous space}.
The group $\grpG$ acting on $\calM$ is called a \emph{symmetry} of $\calM$.

Note that $\phi_{X}$ is a diffeomorphism, with inverse $ \phi_{X}^{-1} = \phi_{X^{-1}}$.
The notation $\Phi_{X} f$ is used to denote the pushforward of the vector field $f$ by $\phi_X$; that is,
\begin{align*}
	\Phi_X f \coloneqq \tD\phi_X \circ f \circ \phi_{X^{-1}}.
\end{align*}

\section{Problem Statement}

In the control problem, it is typical to develop a left-invariant error (for example, \cite{2010_Lee_Geometric}).
As in \cite{Hampsey2022}, we continue in this fashion, noting that the results generalise to right-handed symmetries with the appropriate substitution of left and right actions and translations.
The following section uses results from (\cite{ARCRAS_Mahony_2022}, \cite{2020_mahony_EquivariantSystems}) that were developed for right-handed symmetries; the results for left-handed symmetries are straightforward analogies of these that will be used without proof.

\subsection{Symmetry and Kinematic Systems}

Given a smooth $m$-dimensional manifold $\calM$ and a finite-dimensional vector space $\vecL$, consider the affine system
\begin{align*}
	\dot{\xi} = f_u(\xi) \coloneqq f_0(\xi) + g(\xi) [u].
\end{align*}
Here, $f_0 \in \gothX(\calM)$, $g$ is a linear map $g : \vecL \to \gothX(\calM)$, and $u \in \vecL$ is an input signal.
A trajectory of the system is a curve $\xi: [0, t_f] \to \calM$ satisfying
\begin{align}
	\dot{\xi} & = f_u(\xi),                                  &
	\xi(0)    & = \xi_0, \label{eq:system_dynamics_manifold}
\end{align}
for some admissable signal $u : [0, t_f] \to \vecL$.

Suppose further that $\calM$ is a homogeneous space with Lie group $\grpG$ acting on $\calM$ via the smooth, transitive left action $\phi:\grpG \times \calM \to \calM$.
Then it can be shown that there exists a smooth function $\Lambda: \calM \times \vecL \to \gothg$ satisfying
\begin{align*}
	\tD\phi_{\xi}\Lambda(\xi, u) = f_u(\xi)
\end{align*}
for all $\xi \in \calM, u \in \vecL$, termed a \emph{lift} (\cite{2020_mahony_EquivariantSystems}).
If $\stab_{\phi}(\xi)$ is trivial or discrete, then $\tD\phi_\xi$ is invertible and the lift is unique.

Given an arbitrary fixed \emph{origin} $\mr\xi \in \calM$, the lift determines an associated \emph{lifted system} on $\grpG$:
\begin{align}
	\dot{X} = \tD \mathrm R_{X}\Lambda(\phi(X, \mr\xi), u), \quad\quad\quad \phi(X(0), \mr \xi) = \xi(0) \label{eq:group_dynamics}
\end{align}
where $X\in\grpG$, $u \in \vecL$.
Any trajectory $X(t)$ satisfying \eqref{eq:group_dynamics} will project down onto a trajectory $\xi(t)$ of \eqref{eq:system_dynamics_manifold} via $\xi(t) = \phi(X(t), \mr\xi)$.
In general, there is no unique preimage of a trajectory $\xi(t)$.

If the system admits a left group action $\psi: \grpG \times \vecL \to \vecL$ such that
\begin{align*}
	\Phi_{X}f_u = f_{\psi_X(u)},
\end{align*}
for all $X \in \grpG, u \in \vecL$ then the system is said to be \emph{equivariant}.
Similarly, if a lift $\Lambda$ satisfies
\begin{align*}
	\Ad_X \Lambda(\xi, u) = \Lambda(\phi(X, \xi), \psi(X, u)),
\end{align*}
then the lift $\Lambda$ is said to be equivariant.
\subsection{Trajectory Tracking}
Given a kinematic system \eqref{eq:system_dynamics_manifold}, a desired trajectory $\xi_d(t): [a, b] \to \calM$ is a curve satisfying $\dot{\xi}_d = f_{u_d}(\xi_d)$, \mbox{$\xi_d(0) = \xi^0_{d}$}, where $u_d(t)$ is a known input signal.
The tracking task is to choose a suitable $u(t)$ that forces the system state $\xi(t)$ to $\xi_d(t)$.
In an optimal control framework, $u(t)$ is chosen so as to optimise a cost functional on the set of feasible trajectories.
Let $\Upsilon(\xi_0, [0, t_f])$ denote the set of all feasible trajectories on $\calM$ that start at the initial value $\xi_0$.

A cost functional is then a mapping $J_{\xi_d, u_d} : \Upsilon \to \R $.
The trajectory pair $(\xi^*, u^*)$ is chosen so as to minimise $J_{\xi_d, u_d}$;
\begin{align*}
	(\xi^*, u^*) = \argmin_{(\xi, u) \in \Upsilon} J_{\xi_d, u_d}(\xi, u).
\end{align*}

\section{Equivariant Regulator (E\lowercase{q}R)}

Let $\calM$ be a homogeneous space with symmetry group $\grpG$, system function $f$ and lift $\Lambda$.
Let $\xi_d$ be a desired trajectory to be tracked, and let $\xi$ be the actual system trajectory.
In the EqR framework, we choose a set of intrinsic error coordinates $\xi_e$ on the manifold and study its dynamics.
\subsection{Error dynamics on $\calM$}
In general, there is no intrinsic error between $\xi$ and $\xi_d$ defined for an arbitrary manifold $\calM$.
If $\calM$ is a homogeneous space, however, then an intrinsic error can be formed via the action of the symmetry group $\grpG$ on $\calM$.
To form this error, the desired trajectory $\xi_d$ must be lifted to a corresponding trajectory $X_d$ on the group $\grpG$.
Choose an arbitrary point $\mr \xi \in \calM$, and choose $X^0_d$ so that $\phi(X^0_d, \mr \xi) = \xi_d$.
Then, the solution of $\dot{X}_d = \tD \mathrm R_{X_d} \Lambda(\phi(X_d, \mr \xi), u_d)$ with initial condition $X_d(0) = X^0_d$ is a trajectory on $\grpG$ that satisfies $\phi(X_d, \mr \xi) = \xi_d$.

The point of this construction is that it allows for the definition of an intrinsic error:
\begin{align}
\xi_e(t) \coloneqq \phi(X_d^{-1}, \xi) \label{eq:error_def}.
\end{align}.

\begin{proposition}
Let $\xi_d$ be a desired trajectory, $\mr \xi$ the origin and $X_d$ be a corresponding lifted trajectory.
Let $\xi$ be the actual system trajectory and let $\xi_e$ be the associated error trajectory, defined as in \eqref{eq:error_def}.
Then $\xi_e = \mr\xi$ if and only if $\xi = \xi_d$.
\end{proposition}
\begin{proof}
First, assume that $\xi_e = \mr \xi$.
Then \mbox{$\mr \xi = \xi_e = \phi(X_d^{-1}, \xi)$}.
Left-multiplying both sides by $X_d$,
\begin{align*}
	\phi(X_d, \mr \xi) &=  \phi(X_d, \phi(X_d^{-1}, \xi))
	=  \phi(X_dX_d^{-1}, \xi)
	= \xi.
\end{align*}
The left-hand side is just $\xi_d$, giving the required result.
For the converse argument, assume that $\xi = \xi_d$. Then
\begin{align*}
	\xi_e = \phi(X_d^{-1}, \xi_d)
	= \phi(X_d^{-1}, \phi(X_d, \mr \xi))
	= \phi(X_d^{-1}X_d, \mr \xi) = \mr \xi,
\end{align*}
as required.
\end{proof}
Thus, regulating $\xi_e$ at $\mr\xi$ is equivalent to the task of driving $\xi \to \xi_d$.
For the input error coordinates, define $\tilde{u} = u - u_d$.

\begin{proposition}
	\label{prop:E_dynamics}
	Let $\xi_e$ be an error trajectory defined by \eqref{eq:error_def}.
	Then the time derivative of $\xi_e$ is given by
	\begin{align}
		\dot{\xi_e} & = \tD \phi_{\xi_e}\Ad_{X_d^{-1}}[\Lambda(\phi(X_d, \xi_e), u_d)-\Lambda(\xi_d, u_d)] \notag \\
		            & \phantom{=} + \Phi_{X_d^{-1}} g(\xi_e)[\tilde{u}], \label{eq:xi_dynamics}
	\end{align}
\end{proposition}
where $\Phi_X g$ is the pushforward of $g$ by $\phi_X$ with $\tilde{u}$ held constant.
\begin{proof}
	By the product rule,
	\begin{align}
		\dot{\xi_e} & = \tD_{X|X_d^{-1}} \phi(X, \xi)\ddt (X_d^{-1}) + \tD_{\zeta | \xi} \phi(X_d^{-1}, \zeta)\ddt \xi, \notag                                                    \\
		            & = \tD_{X|X_d^{-1}} \phi(X, \xi)(- \tD L_{X_d^{-1}}\Lambda(\xi_d, u_d)) \notag                                                                               \\
		            & \phantom{=} + \tD_{\zeta | \xi} \phi(X_d^{-1}, \zeta) f_u(\xi), \notag                                                                                      \\
		            & = - \tD_{X|X_d^{-1}} \phi(X, \xi)\tD L_{X_d^{-1}}\Lambda(\xi_d, u_d) \notag                                                                                 \\
		            & \phantom{=} +  \tD_{\zeta | \xi} \phi(X_d^{-1}, \zeta) f_{u_d}(\xi) +  \tD_{\zeta | \xi} \phi(X_d^{-1}, \zeta) g(\xi)[\tilde{u}], \label{eq:error_dyn_mid1} \\
		            & = -\tD_{X|I} \phi(X, \xi_e) \Ad_{X_d^{-1}} \Lambda(\xi_d, u_d) \label{eq:error_dyn_mid2}                                                                    \\
		            & \phantom{=} + \tD_{\zeta | \xi} \phi(X_d^{-1}, \zeta) \tD_{X|I} \phi(X, \xi) \Lambda(\xi, u_d) \label{eq:error_dyn_mid3}                                    \\
		            & \phantom{=} + \tD_\xi \phi(X_d^{-1}, \xi)g(\xi)[\tilde{u}], \notag                                                                                          \\
		            & = \tD_{X|I} \phi(X, \xi_e) \Ad_{X_d^{-1}} (\Lambda(\xi, u_d) - \Lambda(\xi_d, u_d)) \label{eq:error_dyn_mid4}                                               \\
		            & \phantom{=} + \Phi_{X_d^{-1}}g(\xi_e)[\tilde{u}] \label{eq:error_dyn_mid5},
	\end{align}
	where \eqref{eq:error_dyn_mid1} follows from the equality
	\begin{align*}
		f_u(\xi) = f_{\tilde{u} + u_d}(\xi) = f_{u_d}(\xi) + g(\xi)[\tilde{u}],
	\end{align*}
	\eqref{eq:error_dyn_mid2} follows from the relationships
	\begin{align*}
		\tD_{X|X_d^{-1}} \phi(X, \xi)\tD L_{X_d^{-1}}[h] & = \dds \big|_{0}\phi(X_d^{-1}e^{sh}, \xi)              \\
		                                                 & = \dds \big|_{0}\phi(X_d^{-1}e^{sh}, \phi(X_d, \xi_e)) \\
		                                                 & = \dds \big|_{0}\phi(X_d^{-1}e^{sh}X_d, \xi_e)         \\
		                                                 & = D_{X|I}\phi(X, \xi_e)\Ad_{X_d^{-1}}[h],
	\end{align*}
	\eqref{eq:error_dyn_mid3} follows from the definition of the lift
	\begin{align*}
		f_{u_d}(\xi) = \tD_{X|I} \phi(X, \xi) \Lambda(\xi, u_d),
	\end{align*}
	where \eqref{eq:error_dyn_mid4} follows from the relationships
	\begin{align*}
		\tD_{\xi} \phi(X_d^{-1}, \xi)\tD\phi_{\xi}[h] & = \dds \big|_{0}\phi(X_d^{-1}, \phi(e^{sh}, \xi)) \\
		                                              & = \dds \big|_{0}\phi(X_d^{-1}e^{sh}X_d, \xi_e)    \\
		                                              & = D_{X|I}\phi(X, \xi_e)\Ad_{X_d^{-1}}[h],
	\end{align*}
	and \eqref{eq:error_dyn_mid5} follows from the definition of $\Phi_{X_d^{-1}}$.
\end{proof}

\begin{corollary}
	\label{prop:E_dynamics_corollary}
	If the system function $f$ is equivariant, then the time derivative of the error state $\xi_e$ is given by
	\begin{align*}
		\dot{\xi_e} = \tD \phi_{\xi_e}[\Lambda(\xi_e, \mr u_d)-\Lambda(\mr \xi, \mr {u}_d)] + g(\xi_e)[\mr{\tilde{u}}],
	\end{align*}
	where $\mr{\tilde{u}} = \psi(X_d^{-1}, \tilde{u}), \mr {u}_d = \psi(X_d^{-1}, u_d).$
\end{corollary}
\begin{proof}
	By equivariance,
	\begin{align*}
		\Ad_{X_d^{-1}}[\Lambda( & \phi(X_d, \xi_e), u_d)]                                           \\
		                        & = \Lambda(\phi(X_d^{-1}, \phi(X_d, \xi_e)), \psi(X_d^{-1}, u_d)), \\
		                        & = \Lambda(\xi_e, \mr u_d),
	\end{align*}
	and
	\begin{align*}
		\Phi_{X_d^{-1}} g(\xi_e)[\tilde{u}] = g(\xi_e)[\mr{\tilde{u}}].
	\end{align*}
	Applying these identities to Proposition \ref{prop:E_dynamics} gives the required result.
\end{proof}

\subsection{Error dynamics in local coordinates}
\label{sec:local_eps_dynamics}
In the standard application of LQR on a non-linear manifold, local charts have to be chosen along the desired trajectory $\xi_d$.
This requires the problem to be solved in a series of different coordinates, leading to discontinuities and numerical conditioning issues when changing between local charts.
In contrast, by centering a coordinate chart on $\mr\xi$, the EqR error dynamics can be expressed in a single chart, independent of the system trajectory.
Moreover, this chart can be chosen to be well conditioned numerically at the point $\mr{\xi}$ where the asymptotic performance is most important.
Let $\chi$ be a coordinate chart centred on $\mr\xi$ (i.e. $\chi(\mr\xi) = 0$).
Define the error coordinates
\begin{align}
\varepsilon \coloneqq \chi(\xi_e) \label{eq:vareps_def}.
\end{align}

\begin{proposition}
	\label{prop:eps_dynamics}
	Given an origin $\mr \xi$ as well as desired and current trajectories, let $\xi_e$ be an error trajectory defined by \eqref{eq:error_def}.
	Let $\chi$ be a coordinate chart centered on $\mr \xi$ and let $\varepsilon$ be the local coordinate representation of $\xi_e$; $\varepsilon = \chi(\xi_e)$.
	Then the first order dynamics of $\varepsilon$ about $\varepsilon = 0, \tilde{u} = 0$ are
	\begin{align}\label{eq:linearisation}
		\dot{\varepsilon} \approx A(t)\varepsilon + B(t)\tilde{u},
	\end{align}
	where
	\begin{align*}
		A(t) & = \tD \chi \tD\phi_{\mr \xi} \Ad_{X_d^{-1}}\tD_{\xi| \xi_d} \Lambda(\xi, u_d)\tD_{\xi|\mr\xi} \phi_{X_d}(\xi)\tD_{\varepsilon | 0} \chi^{-1}(\varepsilon),
		\\
		B(t) & = \Phi_{X_d^{-1}} g(\mr \xi).
	\end{align*}
\end{proposition}

\begin{proof}
	From Proposition \ref{prop:E_dynamics}, it follows that the full dynamics of $\varepsilon$ are given by
	\begin{align}
		 \dot \varepsilon & = \tD\chi \tD_{X|I} \phi(X, \chi^{-1}(\varepsilon))\Ad_{X_d^{-1}}[-\Lambda(\xi_d, u_d) \notag \\
		    & \quad +\Lambda(\phi(X_d, \chi^{-1}(\varepsilon)), u_d)] \notag                                \\
		    & \quad + \Phi_{X_d^{-1}} g(\chi^{-1}(\varepsilon))[\tilde{u}]. \label{eq:eps_lin_dynamics}
	\end{align}
The matrix $A(t)$ is the linearisation of the $\dot{\varepsilon}$ dynamics \eqref{eq:eps_lin_dynamics} with respect to the state variable $\varepsilon$ evaluated at $(\varepsilon,\tilde{u}) = (0,0)$.
One has
\begin{align}
A(t) & =  \left. \dds \right\vert_{0} \tD\chi \tD_{X|I} \phi(X, \chi^{-1}(s \varepsilon))\Ad_{X_d^{-1}}[-\Lambda(\xi_d, u_d)
	\notag                                                                                                                                                               \\ &\hspace{2cm}
	+\Lambda(\phi(X_d, \chi^{-1}(s \varepsilon)), u_d)]  \notag                                                \\
	 & =  \left. \dds \right\vert_{0}  \tD\chi \tD_{X|I} \phi(X, \chi^{-1}(s \varepsilon))\Ad_{X_d^{-1}}[-\Lambda(\xi_d, u_d)
	\notag                                                                                                                                                               \\ &\quad
	+\Lambda(\phi(X_d, \chi^{-1}(0)), u_d)]
	\notag                                                                                                                                                               \\
	&\quad + \left. \dds \right\vert_{0} \tD\chi \tD_{X|I} \phi(X, \chi^{-1}(0))\Ad_{X_d^{-1}}[-\Lambda(\xi_d, u_d)
	\notag                                                                                                                                                               \\ &\quad
	+\Lambda(\phi(X_d, \chi^{-1}(s \varepsilon)), u_d)]\label{eq:eps_dyn_mid1}                                                                                          \\
	 & = \left. \dds \right\vert_{0}  \tD\chi \tD_{X|I} \phi(X, \chi^{-1}(0))\Ad_{X_d^{-1}}[-\Lambda(\xi_d, u_d)
	\notag                                                                                                                                                               \\ &\hspace{2cm}
	+\Lambda(\phi(X_d, \chi^{-1}(s \varepsilon)), u_d)]\label{eq:eps_dyn_mid2}                                                                                          \\
	 & = \tD \chi \tD\phi_{\mr \xi} \Ad_{X_d^{-1}}\tD_{\xi| \xi_d} \Lambda(\xi, u_d)\tD_{\xi|\mr\xi} \phi_{X_d}(\xi)\tD_{\varepsilon | 0} \chi^{-1}(\varepsilon)
\notag,
\end{align}
	where
	\eqref{eq:eps_dyn_mid1} follows from the product rule and \eqref{eq:eps_dyn_mid2} follows from the identity $\phi(X_d, \chi^{-1}(0)) = \phi(X_d, \mr\xi) = \xi_d$,
	so
	\begin{align*}
		-\Lambda(\xi_d, u_d) +\Lambda(\phi(X_d, \chi^{-1}(0)), u_d) = 0.
	\end{align*}
Note that only the second term in $\dot{\varepsilon}$ \eqref{eq:eps_lin_dynamics} depends on $\tilde{u}$.
Moreover, this dependence is linear.
It follows that the matrix $B(t)$ is
	\begin{align*}
		B(t) = \Phi_{X_d^{-1}} g(\chi^{-1}(\varepsilon)) |_{\varepsilon = 0} = \Phi_{X_d^{-1}} g(\mr \xi).
	\end{align*}
\end{proof}

\begin{corollary}
	\label{prop:eps_dynamics_equiv}
	If $f$ is equivariant, then the $A(t)$ and $B(t)$ matrices in Proposition \ref{prop:eps_dynamics} simplify to
	\begin{subequations}\label{eq:equivariant_linearisation}
		\begin{align*}
			A(t) & = \tD \chi \tD\phi_{\mr \xi} \tD_{\xi| \mr \xi} \Lambda(\xi, \mr{u}_d)\tD_{\varepsilon | 0} \chi^{-1}(\varepsilon),
			\\
			B(t) & = g(\mr \xi) \tD \psi_{X_d^{-1}}.
		\end{align*}
	\end{subequations}
\end{corollary}

\subsection{EqR Design}
\label{sec:eqr_design}

The main contribution of this paper is the development of the Equivariant Regulator.
Once the local linearised error dynamics \eqref{eq:linearisation} have been determined, the EqR is obtained by applying an LQR design to stabilise the error trajectory about the origin.
That is, the control input $\tilde{u}$ is chosen so as to minimise the cost functional
\begin{align}
	J_{\varepsilon, \tilde{u}} \coloneqq \varepsilon(t_f)^\top \mr F \varepsilon(t_f) + \int_{0}^{t_f} \varepsilon^\top \mr Q \varepsilon + \tilde{u}^\top S\tilde{u} 	  \td \tau \label{eq:eqr_cost},
\end{align}
where $t_f \in \R$ is a finite-time horizon,  $\mr F$ and $\mr Q$ are positive semi-definite matrices and $S$ is a positive definite matrix.
The matrices $\mr Q$ and $S$ may also be time-varying.

The cost functional \eqref{eq:eqr_cost} for the system dynamics \eqref{eq:linearisation} is minimised with the input (\cite{anderson2007optimal})
\begin{align}
	\tilde{u} = -K\varepsilon \label{eq:optimal_u},
\end{align}
where $K = S^{-1}B^\top P$ and $P(t)$ is the solution of the Riccati differential equation
\begin{align}
	 \dot{P} &= 	-A^\top P - PA +PBS^{-1}B^\top P - \mr Q, & P(t_f) = \mr F(t_f). \label{eq:LQR_gain}	
\end{align}
The input $u \coloneqq \tilde{u} + u_d$ is then applied as an input to the original system.
\begin{remark}
	The solution $P(t)$ of \eqref{eq:LQR_gain} always exists for the finite time-horizon cost functional posed in \eqref{eq:eqr_cost} \cite[pg. 24]{anderson2007optimal}.
	In the infinite time-horizon case, complete controllability of the linearised error system \eqref{eq:linearisation} is sufficient for the existence of $P(t)$ \cite[pg. 36-37]{anderson2007optimal}.
\end{remark}

\begin{algorithm}
	\caption{EqR Design Preliminaries}
	\begin{algorithmic}[1]
		\Statex Given a system $f: \vecL \to \gothX(\calM)$ and trajectory $\xi_d(t)$;
		\State Find a Lie group with group action $\phi : \grpG \times \calM \to \calM$.
		\State Check for equivariance; if found, compute the input action $\psi: \grpG \times \vecL \to \vecL$.
		\State Compute an (equivariant) lift function $\Lambda: \calM \times \vecL \to \gothg$.
		\State Choose an origin $\mr \xi$ and a local chart $\chi$.
		\State Compute the lifted kinematics $\dot{X}$ and the lifted trajectory $X_d(t)$ (Eq~\ref{eq:group_dynamics}).
		\State Compute the linear matrices $A(t)$ and $B(t)$ (Eq~\ref{eq:linearisation}).
		\State Choose the LQR cost functional matrices $\mr F$, $\mr Q$, $S$.
		\State Precompute the LQR gain $K(t)$ by solving the Riccati equation (Eq~\ref{eq:LQR_gain}).
	\end{algorithmic}
	\label{alg:eqr_algorithm}
\end{algorithm}

Algorithm \ref{alg:eqr_algorithm} summarises the EqR design method.
The gain $K(t)$ does not depend on $\xi$ and so can be solved for ahead-of-time.
The EqR operates on the local error states, so at every time step the control signal is computed from the LQR gain $K$ and $\varepsilon$ (Algorithm \ref{alg:time_step}).

\begin{algorithm}
	\caption{EqR Controller}
	\begin{algorithmic}[1]
		\Statex At time $t$,
		\State  $\xi_e(t) = \phi(X_d^{-1}(t), \xi(t))$ (Eq~\ref{eq:error_def}).
		\State  $\varepsilon(t) = \chi(\xi_e(t))$ (Eq~\ref{eq:vareps_def}).
		\State $\tilde{u}(t) = -K(t)\varepsilon(t)$ (Eq~\ref{eq:optimal_u}).
		\State $u(t) = u_d(t) + \tilde{u}(t)$.
	\end{algorithmic}
	\label{alg:time_step}
\end{algorithm}
\section{Example}
As an illustrative example of the approach, consider the reduced attitude (\cite{Chaturvedi2011}) problem, extended to the control of an underactuated flying robot constrained to thrust in a single direction.
This is an interesting example, as directional thrust control is naturally posed on the sphere.
By considering only the reduced attitude along with the usual thrusting rigid body kinematics, the entire system state can be posed on the manifold $\Sph^2 \times \R^3 \times \R^3$ (which is a homogeneous space, but not a Lie group),
\begin{subequations}\label{eq:s2_dynamics}
	\begin{align}
		\dot{\eta} & = \eta \times \Omega       \\
		\dot{v}    & = -\frac{T}{m}\eta + g e_3 \\
		\dot{x}    & = v.
	\end{align}
\end{subequations}

where the states are the reduced attitude $\eta \in \Sph^2$, the velocity $v \in \R^3$ and the position $x \in \R^3$.
The inputs are the body-frame angular velocity $\Omega \in \R^3$ and the thrust $T \in \R$.
The gravity $g \in \R$ and mass $m \in \R$ are known.

The manifold $\Sph^2 \times \R^3 \times \R^3$ is a homogeneous space under action by the symmetry group $\SE_2(3)$.
For details on the Lie group $\SE_2(3)$, the reader is referred to \cite{2014_Bonnabel}.
Let $(R_X, v_X, x_X)$ denote the components of the group element $X \in \SE_2(3)$.
The group $\SE_2(3)$ acts on $\Sph^2 \times \R^3 \times \R^3$ via the map $\phi: \SE_2(3) \times (\Sph^2 \times \R^3 \times \R^3) \to \Sph^2 \times \R^3 \times \R^3$
defined by
\begin{align*}
	\phi((R_X, v_X, x_X), (\eta, v, x)) = (R_X \eta, R_X v + v_X, R_X x + x_X).
\end{align*}

A lift $\Lambda : \Sph^2 \times \R^3 \times \R^3 \times \vecL \to \se_2(3)$ can be computed as
\begin{align*}
	\Lambda(\xi, u) \coloneqq (-\Omega^{\times},  \quad \Omega \times v - \frac{T}{m}\eta + ge_3, \quad \Omega\times x + v).
\end{align*}
Define the stereographic projection map $\sigma: \Sph^2 \to \R^2$ by $\phi(\eta_1, \eta_2, \eta_3) = (\frac{\eta_1}{\eta_3 + 1}, \frac{\eta_2}{\eta_3 + 1})$.
Construct a chart $\chi :\Sph^2 \times \R^3 \times \R^3 \to \R^8$ by the map $(\eta, v, x) \mapsto (\sigma(\eta), v, x)$.
With this choice of coordinates, a natural choice of origin is $\mr \xi = (\mr\eta, \mr v, \mr x) = (e_3, 0, 0)$.

\subsection{Lifting $\xi_d$ onto $\SE_2(3)$}

The dynamics of $X$ are then given by the lifted system
\begin{align*}
	\dot{X} = (-\Omega^{\times}R_X,  - \frac{T}{m}\phi(R_X, \mr\eta) + ge_3,  \quad \phi(v_X, \mr v)).
\end{align*}
Given an arbitrary feasible trajectory pair $(\xi_d(t), u_d(t))$, a corresponding trajectory on $\SE_2(3)$ must be found.
This can be achieved by finding an $X_d(0)$ such that \mbox{$\phi_{X_d(0)}(\mr{\xi}) = \xi_d(0)$} and then integrating the lifted dynamics with the known input $u_d(t)$.

\subsection{Linearised error dynamics}

Following \S\ref{sec:local_eps_dynamics}, the linearised local error dynamics are computed to be
\begin{align}
	\dot{\varepsilon} & = \begin{pmatrix}
		0 & 0 & 0 \\ \frac{-T_d}{m}\begin{pmatrix} 2 & 0\\ 0 & 2 \\ 0 & 0 \end{pmatrix} & \Omega_d^\times & 0 \\ 0 & \Id & \Omega_d^\times
	\end{pmatrix} \varepsilon +  \begin{pmatrix} \begin{pmatrix}0 & -\frac{1}{2} & 0 \\ \frac{1}{2} & 0 & 0  \end{pmatrix} & 0 \\ 0 & -\frac{1}{m} e_3 \\ 0 & 0 \end{pmatrix} \tilde{u}. \label{eq:linearised_s2}
\end{align}

\vspace{-2mm}
\section{Simulation}
\vspace{-2mm}
In order to verify the performance of the EqR, we compare the tracking performance in simulation to the standard technique of using reprojected error coordinates as input into an LQR.
This is a common technique technique, viewable, for example, in \cite{2018_Scaramuzza_onboard_LQR}, where it was applied to a quaternion state.
This controller is denoted by P-LQR in the remainder of the paper.
\subsection{Trajectory generation on $\Sph^2 \times \R^3 \times \R^3$}
The kinematics \eqref{eq:s2_dynamics} are clearly differentially flat in the position $x_d$ due to the directional thrust constraint.
Let $x_d$ be an arbitrary $\mathcal{C}^3$ curve in $\R^3$. Then $v_d = \dot{x}_d$, $\dot{v}_d = \ddot{x}_d$, $\ddot{v}_d = \dddot{x}_d$.
Thrust is a scalar value, so $T = m\lVert(-\dot{v} + ge_3) \rVert$, $\dot{T} = m\frac{(\ddot{v}^\top(\dot{v} - ge_3))}{T}$, $\eta_d = \frac{m(-\dot{v} + ge_3)}{T}$ and $\dot{\eta}_d = m\dot{T}\frac{\dot{v} - ge_3}{T^2}-m\frac{\ddot{v}}{T}$.
The angular velocity $\Omega_d$ must be chosen so that $\dot{\eta}_d = \eta_d \times \Omega_d$; this is underdetermined and $\Omega_d = \dot{\eta}_d \times \eta_d + \beta(\eta_d, \dot{\eta}_d) \eta_d$ is a solution for any arbitrary function $\beta: \Sph^2 \times T\Sph^2 \to \R$.
We will choose $\beta = 0$ in the following development.

\subsection{Linearised system used for P-LQR comparison}
\label{sec:ambient_error}
\vspace{-2mm}
The error coordinates for P-LQR are the element-wise differences of the state $\tilde{\xi} \in \R^9$, defined by
\begin{align*}
	\tilde{\xi} = (\eta - \eta_d, v - v_d, x - x_d),
\end{align*}

The projection operator $\mathbb{P}_{\eta}: \R^3 \to \tT_{\xi}\Sph^2$ is defined by $\mathbb{P}_\eta = \Id - \eta \eta^\top$, so the projection operator \mbox{$\mathbb{P}: \R^9 \to T_{(\eta, v, x)} (\Sph^2 \times \R^3 \times \R^3)$} is defined by
\begin{align*}
	\mathbb{P} =
	\begin{pmatrix}
		\Id - \eta_d \eta_d^\top & 0   & 0   \\
		0                        & \Id & 0   \\
		0                        & 0   & \Id
	\end{pmatrix}
\end{align*}
Note that $\pr^\top = \pr$.
Thus, the linearised system is:
\begin{align}
	\dot{\tilde{\xi}}
	 & = \tD_{\xi | \xi_d} f_{u_d}\mathbb{P}[\tilde{\xi}] + \tD_{u | u_d} f_{\xi_d, u}[\tilde{u}] \notag \\
	 & =\begin{pmatrix}
		\Omega_d^\times (\Id - \eta_d\eta_d^\top) & 0 & 0 \\
		-\frac{T_d}{m}(\Id - \eta_d\eta_d^\top)   & 0 & 0 \\
		0                                         & I & 0\end{pmatrix}\tilde{\xi} +
	\begin{pmatrix}
		\eta_d^\times & 0                 \\
		0             & -\frac{\eta_d}{m} \\
		0             & 0
	\end{pmatrix}\tilde{u}. \label{eq:proj_system}
\end{align}
\subsection{LQR design and EqR parameter design}
We choose $m = 1.2$ for the simulation.
The system \eqref{eq:proj_system} is tracked with a standard LQR design; that is, the cost functional to be minimised is
\begin{align}
(\xi(t_f) - &\xi_d(t_f))^\top \pr F \pr (\xi(t_f) - \xi_d(t_f)) \notag \\
&+ \int_0^{t_f} (\xi - \xi_d)^\top \pr Q \pr (\xi - \xi_d) + \tilde{u}^\top S \tilde{u} \td\tau,\label{eq:PLQR_cost_functional}
\end{align}
where $\xi$ and $\xi_d$ are treated as embedded coordinates in $\R^9$
and the projection $\pr$ ensures that the cost is well conditioned as a cost on the reduced attitude problem.

The $F$ and $Q$ matrices are chosen to be
\begin{align*}
	F = Q = \diag(1.0, 1.0, 1.0, 2.0, 2.0, 2.0, 0.1, 0.1, 0.1),
\end{align*} and the $S$ matrix is chosen as $S = \diag(0.5, 0.5, 0.5, 0.5)$.

For a fair comparison between P-LQR and EqR performance, the EqR weights must be chosen to represent an equivalent infinitesimal cost to that used in the cost functional \eqref{eq:PLQR_cost_functional}.

One has
\begin{align*}
\pr(\xi - \xi_d) &= \pr(\phi(X_d, \chi^{-1}(\varepsilon)) - \xi_d)\\
&= \tD_{\varepsilon|0}\pr(\phi(X_d, \chi^{-1}(\varepsilon)) - \xi_d)[\varepsilon] + \cal O(\varepsilon ^ 2)\\
&= \pr\tD\phi_{X_d}\tD\chi^{-1}_{|0}[\varepsilon] + \cal O(\varepsilon ^ 2),
\end{align*}
so
\begin{align*}
	(\xi - \xi_d)^\top \pr Q \pr &(\xi - \xi_d)\\
	&\approx \varepsilon^\top   (\tD\chi^{-1}_{|0})^\top \tD\phi_{X_d}^\top \pr Q \pr\tD\phi_{X_d}\tD\chi^{-1}_{|0}\varepsilon.
\end{align*}
Since $\pr \tD\phi_{X_d} = \tD\phi_{X_d}$, it follows that
\begin{align*}
\mr Q & = (\tD\chi^{-1}_{|0})^\top \tD\phi_{X_d}^\top Q \tD\phi_{X_d}\tD\chi^{-1}_{|0}, \\
\mr F & = (\tD\chi^{-1}_{|0})^\top \tD\phi_{X_d}^\top F \tD\phi_{X_d}\tD\chi^{-1}_{|0}.
\end{align*}
Note that while $Q$ is a constant matrix, the EqR weight matrix $\mr Q$ is time-varying.

\subsection{Helix trajectory and transient response}
A helix is a commonly used simulation trajectory, requiring a low-frequency but non-constant angular rate to follow correctly.
The goal of this simulation is to investigate the transient tracking tracking response of a helical trajectory with initial condition perturbation.
The helix trajectory is defined by $x_d(t) = (\frac{1}{2}\cos t, \frac{1}{2}\sin t, t)$.
The initial bearing of the desired trajectory is (0.0509, 0, 0.999); for the simulation, $\eta(0)$ is iteratively selected over $\Sph^2$.
The states $x(0)$ and $v(0)$ are set to $x_d(0)$ and $v_d(0)$, respectively.
The results are shown in  Figure~\ref{fig:EQR_heat} and discussed in \S\ref{sec:discussion}.

\begin{figure}	\includegraphics[width=0.5\linewidth]{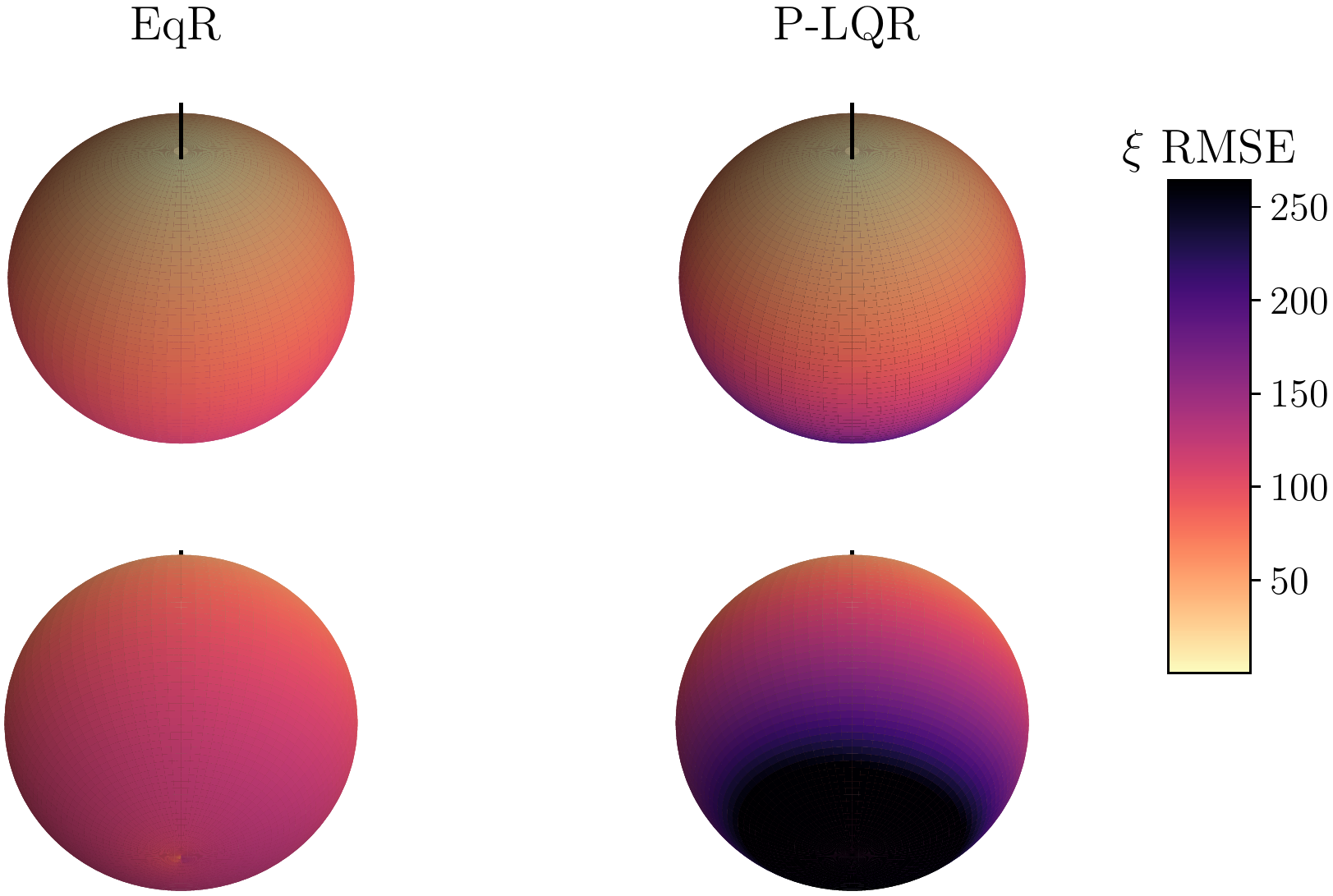}
	\centering
	\caption{Heat map showing total RMSE over the closed-loop trajectory versus initial condition of the true system on the sphere for EqR and P-LQR algorithms.
The desired initial system bearing is denoted by a black line.
Initial conditions far from the desired initial bearing lead to large transients and correspondingly large RMSE.
}	\label{fig:EQR_heat}
\end{figure}

To capture the differences in transient response for large initial error we plot the transient responses for a sample offset of $\theta~=~3.0, \phi~=~1.6$ (Figure~\ref{fig:helix_track_sample}).

\begin{figure}	\includegraphics[width=0.3\linewidth]{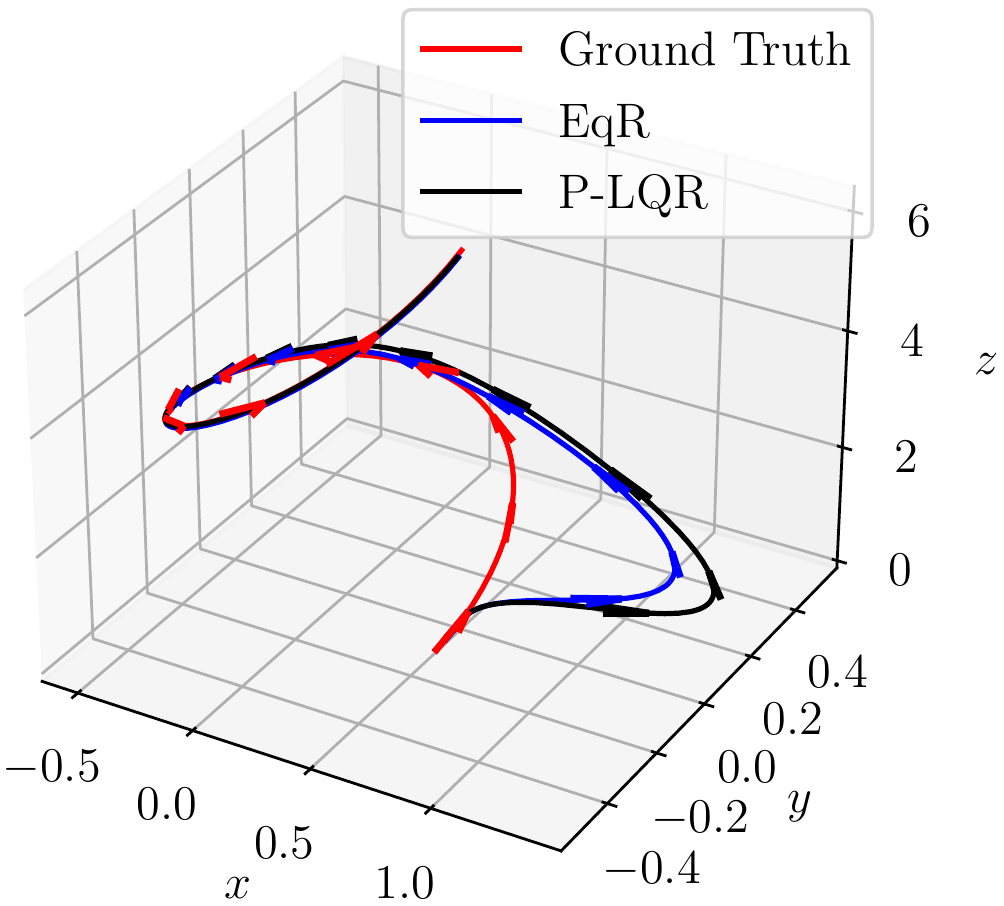}
	\centering
	\caption{Typical trajectory for large initial bearing error.
Note the large transient in both algorithms as the control action acts to correct the initial error and track the helix.
The EqR clearly outperforms the P-LQR algorithm.
}
	\label{fig:helix_track_sample}
\end{figure}

\subsection{Discussion}
\label{sec:discussion}

As is seen in Figure~\ref{fig:EQR_heat} both controllers show similar performance for small initial error.
Greater differences are seen when the initial bearing error is large.
The sample trajectory plotted in Figure~\ref{fig:helix_track_sample} shows a typical response, where the EqR demonstrates superior transient response to the P-LQR algorithm.
EqR was found to converge for all initial perturbations (except for anti-nodal point $\eta_e = (0, 0, -1)$, a chart singularity), whereas P-LQR often failed to converge for large initial $\eta$ values.

Contrasting the linearised systems \eqref{eq:linearised_s2} and \eqref{eq:proj_system}, the EqR system is significantly simpler: it is ``almost" constant, as the system is time-varying only in the inputs $\Omega_d$ and $T_d$, whereas the P-LQR system is also time-varying in the state $\eta_d$.
Thus, EqR is more computationally efficient for this example - the $B$ matrix only needs to be computed once.
This is due to substituted goal in EqR of driving the error trajectory to the fixed point $\mr \xi$ and so computational efficiency should be a general feature of this approach.

\section{Conclusion}

We have shown that, given an input-affine system on a homoegeneous space, a choice of origin allows for a system trajectory to be lifted into a trajectory on the symmetry group.
The group action of such a lifted trajectory on the system state allows for an intrinsic error definition with respect to a feasible trajectory on the manifold.
We have utilised this error to propose the Equivariant Regulator (EqR) as a general tracking controller design for such systems.
The EqR does not require the system manifold itself to be a Lie group and so applies to a broader class of problems than previous approaches.
A simple example of a system on a homogeneous space that is inadmissable for the usual Lie group method has been presented and shows favourable performance compared to a standard LQR tracking controller.

\bibliographystyle{plainnat}
\bibliography{references}

\begin{thebibliography}{20}
\providecommand{\natexlab}[1]{#1}
\providecommand{\url}[1]{\texttt{#1}}
\expandafter\ifx\csname urlstyle\endcsname\relax
  \providecommand{\doi}[1]{doi: #1}\else
  \providecommand{\doi}{doi: \begingroup \urlstyle{rm}\Url}\fi

\bibitem[Anderson and Moore(2007)]{anderson2007optimal}
B.D.O. Anderson and J.B. Moore.
\newblock \emph{Optimal Control: Linear Quadratic Methods}.
\newblock Dover Books on Engineering. Dover Publications, 2007.
\newblock ISBN 9780486457666.

\bibitem[Barrau and Bonnabel(2015)]{2014_Bonnabel}
Axel Barrau and Silvère Bonnabel.
\newblock Invariant filtering for pose ekf-slam aided by an imu.
\newblock In \emph{2015 54th IEEE Conference on Decision and Control (CDC)},
  pages 2133--2138, 2015.
\newblock \doi{10.1109/CDC.2015.7402522}.

\bibitem[Chaturvedi et~al.(2011)Chaturvedi, Sanyal, and
  McClamroch]{Chaturvedi2011}
Nalin~A. Chaturvedi, Amit~K. Sanyal, and N.~Harris McClamroch.
\newblock Rigid-body attitude control.
\newblock \emph{IEEE Control Systems Magazine}, 31\penalty0 (3):\penalty0
  30--51, 2011.
\newblock \doi{10.1109/MCS.2011.940459}.

\bibitem[Cohen et~al.(2020)Cohen, Abdulrahim, and
  Forbes]{2020_Forbes_Quadrotor}
Mitchell~R. Cohen, Khairi Abdulrahim, and James~Richard Forbes.
\newblock {{Finite-Horizon}} {{LQR}} {{Control}} of {{Quadrotors}} on
  {{$SE_2(3)$}}.
\newblock \emph{IEEE Robotics and Automation Letters}, 5\penalty0 (4):\penalty0
  5748--5755, 2020.
\newblock \doi{10.1109/LRA.2020.3010214}.

\bibitem[Farrell et~al.(2019)Farrell, Jackson, Nielsen, Bidstrup, and
  McLain]{2019_Farrell_Error}
Michael Farrell, James Jackson, Jerel Nielsen, Craig Bidstrup, and Tim McLain.
\newblock {{Error-State LQR Control of a Multirotor UAV}}.
\newblock In \emph{2019 International Conference on Unmanned Aircraft Systems
  (ICUAS)}, pages 704--711, 2019.
\newblock \doi{10.1109/ICUAS.2019.8798359}.

\bibitem[Foehn and Scaramuzza(2018)]{2018_Scaramuzza_onboard_LQR}
Philipp Foehn and Davide Scaramuzza.
\newblock {{Onboard State Dependent LQR for Agile Quadrotors}}.
\newblock In \emph{2018 IEEE International Conference on Robotics and
  Automation (ICRA)}, pages 6566--6572, 2018.
\newblock \doi{10.1109/ICRA.2018.8460885}.

\bibitem[Hampsey et~al.(2022)Hampsey, van Goor, Hamel, and Mahony]{Hampsey2022}
Matthew Hampsey, Pieter van Goor, Tarek Hamel, and Robert Mahony.
\newblock Exploiting different symmetries for trajectory tracking control with
  application to quadrotors, 2022.
\newblock URL \url{https://arxiv.org/abs/2207.04782}.

\bibitem[Johnson and Beard(2021)]{Johnson2021}
Jacob Johnson and Randal Beard.
\newblock Globally-attractive logarithmic geometric control of a quadrotor for
  aggressive trajectory tracking, 2021.
\newblock URL \url{https://arxiv.org/abs/2109.07025}.

\bibitem[Lee(2012)]{2012_lee_SmoothManifolds}
John~M. Lee.
\newblock \emph{Introduction to Smooth Manifolds}.
\newblock Graduate Texts in Mathematics. Springer, 2012.

\bibitem[Lee et~al.(2010)Lee, Leok, and McClamroch]{2010_Lee_Geometric}
Taeyoung Lee, Melvin Leok, and N.~Harris McClamroch.
\newblock {{Geometric tracking control of a quadrotor UAV on $SE(3)$}}.
\newblock In \emph{49th IEEE Conference on Decision and Control (CDC)}, pages
  5420--5425, 2010.
\newblock \doi{10.1109/CDC.2010.5717652}.

\bibitem[Mahony et~al.(2020)Mahony, Hamel, and
  Trumpf]{2020_mahony_EquivariantSystems}
Robert Mahony, Tarek Hamel, and Jochen Trumpf.
\newblock Equivariant {{Systems Theory}} and {{Observer Design}}.
\newblock \emph{arXiv:2006.08276 [cs, eess]}, August 2020.

\bibitem[Mahony et~al.(2022)Mahony, van Goor, and Hamel]{ARCRAS_Mahony_2022}
Robert Mahony, Pieter van Goor, and Tarek Hamel.
\newblock Observer design for nonlinear systems with equivariance.
\newblock \emph{Annual Review of Control, Robotics, and Autonomous Systems},
  5\penalty0 (1):\penalty0 221--252, may 2022.
\newblock \doi{10.1146/annurev-control-061520-010324}.

\bibitem[Meyer(1971)]{Meyer1971}
George Meyer.
\newblock Design and global analysis of spacecraft attitude control systems.
\newblock 1971.

\bibitem[Poincar{\'e}(1885)]{poincare1885courbes}
Henri Poincar{\'e}.
\newblock Sur les courbes d{\'e}finies par les {\'e}quations
  diff{\'e}rentielles.
\newblock \emph{J. Math. Pures Appl.}, 4:\penalty0 167--244, 1885.

\bibitem[Slotine and Sastry(1983)]{slotine_1983_trajectory}
J.~J. Slotine and S.~S. Sastry.
\newblock Tracking control of non-linear systems using sliding surfaces, with
  application to robot manipulators†.
\newblock \emph{International Journal of Control}, 38\penalty0 (2):\penalty0
  465--492, 1983.
\newblock \doi{10.1080/00207178308933088}.
\newblock URL \url{https://doi.org/10.1080/00207178308933088}.

\bibitem[Suicmez and Kutay(2014)]{suicmez2014optimal}
Emre~Can Suicmez and Ali~Turker Kutay.
\newblock Optimal path tracking control of a quadrotor uav.
\newblock In \emph{2014 International Conference on Unmanned Aircraft Systems
  (ICUAS)}, pages 115--125. IEEE, 2014.

\bibitem[Tu(2010)]{tu2010introduction}
L.W. Tu.
\newblock \emph{An Introduction to Manifolds}.
\newblock Universitext. Springer New York, 2010.
\newblock ISBN 9781441973993.

\bibitem[van Goor et~al.(2022)van Goor, Hamel, and Mahony]{EQF2022}
Pieter van Goor, Tarek Hamel, and Robert Mahony.
\newblock Equivariant filter (eqf).
\newblock \emph{IEEE Transactions on Automatic Control}, pages 1--13, 2022.
\newblock \doi{10.1109/TAC.2022.3194094}.

\bibitem[Wie and Barba(1985)]{Wie1985}
Bong Wie and Peter~M. Barba.
\newblock Quaternion feedback for spacecraft large angle maneuvers.
\newblock \emph{Journal of Guidance, Control, and Dynamics}, 8\penalty0
  (3):\penalty0 360--365, 1985.
\newblock \doi{10.2514/3.19988}.

\bibitem[Yang and Kim(1999)]{Yang_1999_sliding_traj}
Jong-Min Yang and Jong-Hwan Kim.
\newblock Sliding mode control for trajectory tracking of nonholonomic wheeled
  mobile robots.
\newblock \emph{IEEE Transactions on Robotics and Automation}, 15\penalty0
  (3):\penalty0 578--587, 1999.
\newblock \doi{10.1109/70.768190}.

\end{thebibliography}

\end{document}